\documentclass{article}[12pt]
\usepackage{indentfirst}
\usepackage{tabularx}
\usepackage[a4paper]{geometry}
\usepackage[utf8]{inputenc}
\usepackage[T1]{fontenc}
\usepackage{lmodern}
\usepackage{amsmath,amssymb,amsthm} \allowdisplaybreaks[1]
\usepackage{mathrsfs}
\usepackage{microtype}
\usepackage{latexsym}
\usepackage[colorlinks]{hyperref}
\usepackage{cleveref}
\usepackage{setspace}
\usepackage{graphicx}
\usepackage{bookmark}
\usepackage{booktabs}
\usepackage{diagbox}
\usepackage{float}
\usepackage{color}   
\usepackage{bm}
\usepackage{amssymb}

\theoremstyle{definition}
\newtheorem{de}{Definition}[section]

\newtheorem{ex}{Example}[section]
\newtheorem{fa}{Fact}[section]

\newcommand{\F}{\mathbb{F}}
\newcommand{\Z}{\mathbb{Z}}
\newcommand{\C}{\mathcal{C}}
\newcommand{\D}{\mathcal{D}}

\title{A quaternary analogue of Tang-Ding codes\thanks{This research is supported by the National Natural Science Foundation of China (12071001) and the
National Research Foundation of Korea (NRF) Grant funded by the Korean government (NRF-2019R1A2C1088676).}}

\author{ Minjia Shi$^\ast$\thanks{smjwcl.good@163.com},
	Sihui Tao \thanks{taosihui2022@163.com}, Jon-Lark Kim\thanks{jlkim@sogang.ac.kr},
	Patrick Sol\'e\thanks{{sole@enst.fr}}
	\thanks{Minjia Shi and Sihui Tao are with the Key Laboratory of Intelligent Computing Signal
		Processing, Ministry of Education, School of Mathematical Sciences, Anhui
		University, Hefei 230601, China; State Key Laboratory of integrated Service Networks, Xidian University, Xi'an,
		710071, China. Jon-Lark Kim is with Department of Mathematics, Sogang University, Seoul, South Korea. Patrick Sol\'e is with Aix Marseille Univ, CNRS, Centrale Marseille, I2M, Marseille, France. }}

\date{}

\makeatletter

\newcommand{\Rmnum}[1]{\expandafter\@slowromancap\romannumeral #1@}
\makeatother

\begin{document}
	\maketitle
	\begin{abstract}
			In a recent paper, Tang and Ding introduced a class of binary cyclic codes of rate close to one half with a designed lower bound on their minimum distance. The definition involves the base $2$ expansion of the integers in their defining set. In this paper we propose an analogue for quaternary codes. In addition, the performances of the subfield subcode and of the trace code (two binary cyclic codes) are investigated.
		
	\end{abstract}
	\textbf{Keywords:} Cyclic codes, Duadic codes, Square root bound, LCD codes\\
	\textbf{MSC(2020):} 94 B15
	
	\section{Introduction}\label{Introduction}
	Cyclic codes form the most studied class of error-correcting codes, due to their elegant polynomial structure, efficient decoding algorithms, and minimum
distance lower bounds such as the BCH bound. Most classical families of codes, like Reed-Muller, Reed-Solomon or Hamming codes are either cyclic or extended cyclic. In particular, there is the class of quadratic residue codes (and more generally duadic codes) whose rate is close to one half, length is $n,$ and which admits a lower bound on the minimum distance of order $\sqrt{n}$ the so-called square root bound (Cf. \cite[\S 6.5]{HufPle}, \cite{duadic codes,Q-codes,SD duadic code}). These families of codes were developed to construct self-dual codes by adding an overall parity-check, which is why their rate is close to one half. In fact, every binary cyclic code whose extension is self-dual is duadic \cite{SD duadic code}. These facts are documented in any good treatise on Coding Theory such as \cite{HufPle,Ling san}.
	Recently, Tang and Ding \cite{TD} introduced a class of cyclic codes of parameters
	$[n,(n+1)/2]$ for $n=2^m-1$ by their generator polynomial
$$
g_{\left( i,m \right)}(x)=\prod_{\begin{array}{c}
		w_2(j)\equiv i\,\,\left( \mathrm{mod}\;2 \right)\\
		1\leqslant j\leqslant n-1\\
\end{array}}{(}x-\alpha ^j).
$$
	Here $\alpha$ denotes a primitive root in $\F_{2^m}$ and $w_2(i)$ the Hamming weight of the binary expansion of the integer $i$. They denote by $\mathcal{C} _{\left( i,m \right)}$
	 such a code for $i \in \{0,1\}.$ Simply using the BCH bound, these authors were able to establish a square-root bound for these codes, similar to that known for quadratic residue codes and duadic codes.
These codes are also of rate close to one half, and are, therefore, direct competitors of the duadic codes.

By using the same method, more cyclic codes with good parameter have been found. In \cite{mod3and4}, Liu, Li and Ding constructed five families of cyclic codes and their duals of length $2^m-1$ over $\F_2$ with a similar generator polynomial
$$
g_{\left( i_1,i_2,m \right)}(x)=\prod_{\begin{array}{c}
		w_2(j)\equiv i_1~{\rm{or}}~i_2\,\,\left( \mathrm{mod}\;4 \right)\\
		1\leqslant j\leqslant 2^m-2\\
\end{array}}{(}x-\alpha ^j),
$$
where they take $i_1$ and $i_2$ differently in $\left\{ 0,1,2,3 \right\} $.
The remaining three families cyclic codes have generator polynomials of the form
$$
g_{\left( i,m \right)}(x)=\prod_{\begin{array}{c}
		w_2(j)\equiv i\,\,\left( \mathrm{mod}\;3 \right)\\
		1\leqslant j\leqslant 2^m-2\\
\end{array}}{(}x-\alpha ^j).
$$
As a result, they consider different generator polynomials by changing the module number. Those five families of cyclic codes with the same length $n=2^m-1$ and dimension $
\left( n-6 \right) /3\leqslant k\leqslant \left( n+6 \right) /3,$ and a good lower bound on their minimum distances which is near the square root bound \cite{HufPle}.

It is worth noting that Sun in \cite{sunzhonghua} constructed two families of cyclic codes and their duals over $\F_2$ that have better parameter than Tang and Ding in \cite{TD}. He claimed the parameters
 $[2^m-1, 2^{m-1}, d\geqslant 3\cdot 2^{\left( m-1 \right) /2}-1]$ where $m$ is odd and $m\geqslant 9$.

For the ternary alphabet, Chen, Ding, Li and Sun constructed four infinite families of cyclic codes and their duals in \cite{q=3}. They fixed $n=3^m-1$ and considered generator polynomials (with   $\alpha $ a primitive element of $\F_{3^m}$)
$$
g_{\left( i_1,i_2,m \right)}(x)=\prod_{\begin{array}{c}
		w_2(j)\equiv i_1~{\rm{or}}~i_2\,\,\left( \mathrm{mod}\;4 \right)\\
		1\leqslant j\leqslant 3^m-2\\
\end{array}}{(}x-\alpha ^j),
$$
where they chose $i_1$ and $i_2$ independently in $\left\{ 0,1,2,3 \right\} $. However, they only considered the case when $m$ is odd and got a dimension of around half of the length but the lower bounds of minimum distance is far below the square root bound.

In this paper, we aim to construct similar cyclic codes over $\mathbb{F}_4,$ of length $4^m-1$ for some integer $m,$ in relation with the four weight $w_4()$ of an integer, defined as the rational sum of its digits in its base $4$ expansion. Not only we could emulate the Tang/Ding approach, but also the arithmetics of $w_2()$ bear on that of $w_4(),$ as will be seen in some cases.
	In order to construct long arithmetic series of indices, suitable for the application of the BCH bound, it will be necessary to discuss on the congruence classes of $m$ modulo $4.$ For $m$ odd or $\equiv 2 \pmod{4}$ or $\equiv 4 \pmod{8}$, we are able to derive a lower bound on the minimum distance of our codes of order square root of the length, up to a multiplicative constant.
When $m$ is odd our codes turn out to be a pair of duadic codes. Our bounds in that case show that certain duadic codes admit a square root bound on the minimum distance (not only on the odd-like distance as in Fact 2.5 below). When $m$ is even our codes form a new family of LCD cyclic codes of rate close to one half. Further, there are two canonical binary codes attached to a quaternary code: the subfield subcode and the trace code.
We will study their parameters as well. We will also study the binary image of the quaternary codes studied along the lines of \cite{type II}, and the extension of our codes by an overall parity-check.

The material is laid out in the following way. The next section collects the basic notions and definitions needed for the other sections.
Section 3 introduces the two families of quaternary cyclic codes we consider. Section 4 contains some technical results. Section 5 and Section 6 give the main results of the paper, that is lower bounds on the minimum distance.
Section 7 studies the properties of extension codes, subfield subcodes and  trace codes. Section 8 concludes the article.
	\section {Preliminaries}
	In this section, we will give some basic concepts and definitions needed throughout the paper.
	\begin{de}
			Let $\mathbb{F}_q$ denote the finite field of order $q$, where $q$ is a prime power.  Let $\mathbb{F}_{q}^{n}$
		denote the vector space of all $n$-tuples over the finite field $\mathbb{F}_q$.  If $\C$ is a $k$-dimensional
		subspace of $\mathbb{F}_{q}^{n}$, then $\mathcal{C}$ will be called an $[n,k]$ {\bf linear code} over $\mathbb{F}_q$. The linear code $\C$ has $q^{k}$ codewords.
		The dual code $\C^{\bot}$ can be defined by
	$$
	\C^{\bot}=\left\{ \boldsymbol{x}\in \mathbb{F} _{q}^{n}\left| \forall \,\,\boldsymbol{c}\in \C, \boldsymbol{x}\cdot \boldsymbol{c}=0 \right. \right\} .
	$$
	
	\end{de}

\begin{de}
	If a code $\C$ satisfies $\C\subseteq \C^{\bot}$, then $\C$ is said to be {\bf self-orthogonal.} If a code $\C$ meets $\C= \C^{\bot}$, then $\C$ is said to be {\bf self-dual.}
	
\end{de}

\begin{de}
	A linear code $\C$ is called an {\bf LCD code} (linear complementary dual code) if $\C\cap \C^{\bot}=\left\{ 0 \right\} $.
\end{de}

\begin{de}
	A subset $S$ of $\mathbb {F}_{q}^{n}$ is {\bf cyclic} if $
	\left( c_{n-1}c_0c_1\cdots c_{n-2} \right) \in \C$ whenever $\left( c_0c_1c_2\cdots c_{n-1} \right) \in \C$. A linear code $\C$ is called a {\bf cyclic code} if $\C$ is a cyclic set. Cyclic codes over $\mathbb{F}_q$ of length $n$ are precisely the ideals of $R_n=\mathbb {F}_q[x]/(x^n-1)$.

\end{de}

\begin{de}
	Let $n$ be coprime to $q$, the {\bf $q$-cyclotomic coset} of $i$ modulo $n$ is defined by
	$$
	C_{i}^{\left( q,n \right)}=\left\{ i\cdot q^l\,\,mod\,\,n:0\leqslant l\leqslant l_{\max}-1 \right\},
	$$ where $l_{max}$
is the least positive integer such that $
i\cdot q^{l_{\max}}\equiv i\,\,\left( mod\,\,n \right)
$, and $l_{max} $ is equal to the size of $C_{i}^{\left( q,n \right)}$.

\end{de}

\begin{fa}
	 Let $\alpha$ be a primitive
	$n_{th}$ root of unity in some extension field of $\F_q$.
	For any cyclic code $\C$ in $R_n$  its generator polynomial $g(x)$ can be written as 	$$
	g(x)=\prod_s{M_{\alpha ^s}\left( x \right) =\prod_{i\in C_{s}^{\left( q,n \right)}}{\left( x-\alpha ^i \right)}},
	$$	where $s$ runs through some subset of
	representatives of the $q$-cyclotomic cosets $C_{s}$ modulo $n$.
	
\end{fa}

\begin{de}
	Let $T=\bigcup_s{C_s}$ be the union of these $q$-cyclotomic cosets.
	 The roots of unity $Z=\left\{ \alpha ^i|i\in T \right\} $
	  are called the {\bf zeros} of the cyclic code $\C$.  The set $T$ is called the {\bf defining set} of $\C$.

\end{de}

\begin{fa}
	
	Let $\C$ be a cyclic code of length $n$ over $\F_q$ with zeros $
	Z=\left\{ \alpha ^i|i\in T \right\}
	$
	 for
	some primitive $n$th root of unity $\alpha$ where $T$ is the defining set of $\C$. Let $a$ be an integer such
	that $\gcd(a, n) = 1$ and let $a^{-1}$ be the multiplicative inverse of $a$ in the integers modulo $n$.
	Then $\left\{ \alpha ^{a^{-1}i}|i\in T \right\} $ are the zeros of the cyclic code $\C_{\mu a} $ and $a^{-1}i$ mod $n$ is the defining
	set for $\C_{\mu a} $. The cyclic code $\C$ and $\C_{\mu a}$ are equivalent.

\end{fa}

\begin{fa}\label{duadic code definition}
	Let $\C_1$ and $\C_2$ be two cyclic codes over $\F_q$ with defining sets $T_1$ and $T_2$, respectively, where $0\notin T_1$ and $0\notin T_2$. Then $\C_1$ and $\C_2$ are a pair of odd-like duadic codes if and only if:
	
(1)	$T_1$ and $T_2$ satisfy $
	T_1\cup T_2=\left\{ 1,2,\cdots ,n-1 \right\}
	$ and $T_1\cap T_2=\oslash $ and
	
(2)	there is a multiplier $\mu _b$ such that $T_1\mu _b=T_2$ and $T_2\mu _b=T_1$.

\end{fa}

\begin{de}
	The {\bf Lee composition} of a vector $
	\boldsymbol{x}=\left( x_1,\cdots ,x_n \right) \in \F_{4}^{n}
	$
	is defined as $$(n_0(\boldsymbol{x}),n_1(\boldsymbol{x}),n_2(\boldsymbol{x}))$$ where $n_0(\boldsymbol{x})$ is the number of $x_i=0$, $n_2(\boldsymbol{x})$ is the number of $x_i=1$, and $n_1(\boldsymbol{x})=n-n_0(\boldsymbol{x})-n_2(\boldsymbol{x})$. The {\bf Lee weight} $w_L(\boldsymbol{x})$ of $\boldsymbol{x}$ is then defined as
	$$n_1(\boldsymbol{x})+2n_2(\boldsymbol{x}).$$
	The {\bf Lee distance} of two codewords $\boldsymbol{x}$ and $\boldsymbol{y}$ is the Lee weight of $\boldsymbol{x}-\boldsymbol{y}$.
	
\end{de}

\begin{de}\cite{type II}
	A self-dual code $\C$ over $\F_4$ is said to be {\bf Type II} if the Lee weight of every codeword is a multiple of $4$ and {\bf Type I} otherwise.
	
\end{de}

\begin{de}
	Let   $\F_4=\left\{ 0,1,w,\overline{w} \right\} $.
	There is a natural Gray map $\phi $ which is a $\F_2$-linear isometry from
	($	\F_{4}^{n}$, lee distance) onto ($\F_{2}^{2n}$, Hamming distance). For all $\boldsymbol{x},\boldsymbol{y} \in \F_{2}^{n}$, the {\bf Gray map} $\phi$ is defined as
	$$
	\phi \left( w\boldsymbol{x}+\overline{w}\boldsymbol{y} \right) =\left( \boldsymbol{x},\boldsymbol{y} \right).
	$$

\end{de}

\begin{fa}(\cite[BCH bound]{HufPle})
	Let $\C$ be a cyclic code of length $n$ over $\mathbb{F} _q$ with defining set
	$T$. Suppose $\C$ has minimum weight $d$. Assume $T$ contains $\delta -1$ consecutive elements for some integer $\delta-1$. Then $d\geqslant \delta $.
\end{fa}

\begin{de} A vector of length $n$ over $\F_q$ is said to be {\bf odd-like} if it does belong to the hyperplane $\{ \boldsymbol{x} \in \F_q^n \mid \sum_{i=1}^n x_i=0\}.$
The {\bf odd-like distance} of a linear code $\subseteq \F_q^n$ is the minimum weight of its odd-like vector.
\end{de}

We quote the following square-root bound for duadic codes. Note that it does not apply to the minimum distance.
\begin{fa} (\cite[Th. 6.5.2]{HufPle}) Let $\D_1$ and $\D_2$ be a pair of odd-like duadic codes of length $n$ over $\F_q$ . Let $d_0$ be their (common) minimum odd-like weight. Then the following hold:
	\begin{enumerate}
		\item [(1)] $d_{0}^{2}\geqslant n$.
		\item [(2)] If the splitting defining the duadic codes is given by $\mu _{-1}$, then $d_{0}^{2}-d_0+1\geqslant n$.
		\item [(3)] Suppose $d_{0}^{2}-d_0+1=n$, where $d_0>2$, and  assume that the splitting defining the duadic codes is given by $\mu =-1$. Then $d_0$ is the minimum weight of both $\D_1$ and $\D_2$.
		
	\end{enumerate}
	
\end{fa}
\hspace{-0.68cm}{\bf Remark:} See \cite[Th. 6.6.22]{HufPle} for the square root bound on quadratic residue codes which does apply to the minimum weight.
\begin{de} \cite{type II}
	Let $\C$ be a cyclic code over $\F_4$. The {\bf subfield subcode} $\C_2$ of $\C$ is the linear subcode of $\C$ consisting of those codewords all entries of which are binary. The {\bf trace code} $Tr(\C)$ of $\C$ is the
	binary code obtained from $\C$ by taking the trace of each codeword coordinatewise.
	
\end{de}

	\section {The First Two Families of Quaternary Cyclic Codes}
	Let $m\geqslant 1$ be a positive integer and let $n=4^m-1$. For any $i$ between $0$ and $4^m-1$, the expression
\begin{equation}
		i=i_{m-1}4^{m-1}+i_{m-2}4^{m-2}+...+i_14+i_0 \notag\end{equation}
 is the $4$-adic expansion of $i$, where $i_j\in \left\{ 0,1,2,3 \right\} $ for $0\leqslant j\leqslant m-1$.
	
The  $4$-weight of $i$, denoted by $w_4(i)$ is defined to be  $\omega _4\left( i \right) =\sum_{j=0}^{m-1}{i_j}.$ Define  $$
	T_{\left( 0,m \right)} =\left\{ 1\leqslant i\leqslant n-1:\omega _4\left( i \right) \equiv 0 \pmod{2} \right\}, $$ and $$
	 T_{\left( 1,m \right)} =\left\{ 1\leqslant i\leqslant n-1: \omega _4\left( i \right) \equiv 1 \pmod{2} \right\}.
	 $$
	
	 Let $\alpha $ be a generator of $\mathbb{F}_{4^m}^*$, then the cyclic code $\C_{(0,m)}$ and $\C_{(1,m)}$ are generated by the defining set $T_{(0,m)}$ and $T_{(1,m)}$, respectively. Actually we define a polynomial

$$
	g_{\left( i,m \right)}\left( x \right) =\prod_{\begin{array}{c}
			1\leqslant j\leqslant n-1\\
			w_4\left( j \right) \equiv i\pmod{2}\\
	\end{array}}^{}{\left( x-\alpha ^j \right)}
	$$
	for each $	i\in \left\{ 0,1 \right\} $.
	Some Frobenius action shows that $g_{\left( i,m \right)}\left( x \right) 	\in {\mathbb{F}}_4\left[ x \right] $. When $m=2$, Magma \cite{M} shows that the parameters of the quaternary cyclic code $\C_{(1,m)}$ are $[15,7,5]$, which is near optimal. This motivated us to study the parameters of two families of quaternary codes and their duals.
	
	\section{Some Auxiliary Results}
In this section, we will present some auxiliary results on the defining sets of the two families of quaternary cyclic codes $\C_{(i,m)}$, which will play an important role in developing good lower bounds on the minimum distance of quaternary codes.

	\newtheorem{mylemma}{Lemma}[section]
	\begin{mylemma}\label{gcd(v,n)=1}
		Let $l$ be any positive integer and let $l/\gcd(m,l)$ be odd. Then $$
		\gcd\left( a^m+1,a^l-1 \right) =\begin{cases}
			1~~~~{\rm{if}}~a~\rm{is~even},\\
			2~~~~{\rm{if}}~a~\rm{is~odd}.\\
		\end{cases}
		$$
	
	\end{mylemma}

\begin{mylemma}
	For any integer $m\geqslant 1$, let $v=4^{s}+1$. If $m/\gcd\left( m,s \right)$ is odd, then the set $
	\left\{ av:1\leqslant a\leqslant 4^x \right\} \subset T_{(0,m)}
	$, where $x$ is a positive integer and meets $s+x<m$.
	
\end{mylemma}

\begin{proof}
	This proof can be divided into four parts as per
 the following lemmas.
\end{proof}

\begin{mylemma}\label{m=0(mod4)}
Let $m\equiv 4 \pmod{8} \geqslant 12$ and $v=4^{\frac{m-4}{2}}+1$. Then $$
	\left\{ av:1\leqslant a\leqslant 4^{\frac{m-4}{2}} \right\} \subset T_{\left( 0,m \right)}.
	$$

\end{mylemma}
\begin{proof}
	In this case, from Lemma \ref{gcd(v,n)=1}, it can be seen that $\gcd\left( v,n \right) =1$. Suppose that $a=4^{{\frac{m-4}{2}} }$, $w_{4}\left( av\right) =w_{4}\left( v\right) =2$, so
	$av\in T_{\left( 0,m \right)}$. Suppose $1\leqslant a\leqslant 4^{{\frac{m-4}{2}} }-1$, write $
	a=\sum_{i=0}^{\frac{m-6}{2}}{a_i4^i}$, where $a_i\in \left\{ 0,1,2,3 \right\} $.
So we have
	\begin{align}
		\nonumber	av=\left( \sum_{i=0}^{\frac{m-6}{2}}{a_i4^i} \right) \left( 4^{\frac{m-4}{2}}+1\right) =\sum_{i=0}^{\frac{m-6}{2}}{a_i4^{i+{\frac{m-4}{2}}}}+\sum_{i=0}^{{\frac{m-6}{2}}}{a_i4^i}.
	\end{align}
	Since the exponents in these two sums are pairwise distincts we have
	\begin{align}
		\nonumber		w_4\left( av \right) =w_4\left( \sum_{i=0}^{\frac{m-6}{2}}{a_i4^{i+{\frac{m-4}{2}}}}+\sum_{i=0}^{{\frac{m-6}{2}}}{a_i4^i} \right)
		=w_4\left( \sum_{i=0}^{\frac{m-6}{2}}{a_i4^{i+{\frac{m-4}{2}}}} \right) +w_4\left( \sum_{i=0}^{\frac{m-6}{2}}{a_i4^i} \right) 
		=2w_4\left( a \right) ,
	\end{align}  so $av\in T_{\left( 0,m \right)}$.
	
\end{proof}

\begin{mylemma}\label{m mod 4 =1}
Let $m\equiv 1 \pmod{4}\geqslant 5$ and $v=4^{\frac{m-1}{2}}+1$. Then $$
	\left\{ av:1\leqslant a\leqslant 4^{\frac{m-1}{2}} \right\} \subset T_{(0,m)}.$$
\end{mylemma}

\begin{proof}
	In this case, by Lemma \ref{gcd(v,n)=1}, we can  get that $\gcd(v,n)=1$.  Suppose that $a=4^{\frac{m-1}{2}}$, $w_4\left( av \right) =w_4\left( v \right) =2$, so $av\in T_{\left( 0,m \right)}$. Suppose $1\leqslant a\leqslant 4^{\frac{m-1}{2}}-1$, write $a=\sum_{i=0}^{\frac{m-3}{2}}{a_i4^i}$, where $a_i\in \left\{ 0,1,2,3 \right\} $.
	So we have
\begin{align}
	\nonumber av&=\left( \sum_{i=0}^{\frac{m-3}{2}}{a_i4^i} \right) \left( 4^{\frac{m-1}{2}}+1 \right) 
=\sum_{i=0}^{\frac{m-3}{2}}{a_i\left( 4^i+4^{i+\frac{m-1}{2}} \right)}
=\sum_{i=0}^{\frac{m-3}{2}}{a_i4^i+\sum_{i=0}^{\frac{m-3}{2}}{a_i4^{i+\frac{m-1}{2}}}}.
\end{align}
It is clear that
\begin{align}
	\nonumber		w_4\left( av \right) &=
	w_4\left( \sum_{i=0}^{\frac{m-3}{2}}{a_i4^{i+\frac{m-1}{2}}}+\sum_{i=0}^{\frac{m-3}{2}}{a_i4^i} \right) 
	=w_4\left( \sum_{i=0}^{\frac{m-3}{2}}{a_i4^{i+\frac{m-1}{2}}} \right) +w_4\left( \sum_{i=0}^{\frac{m-3}{2}}{a_i4^i} \right) 
	=2w_4\left( a \right) ,
\end{align}  so $av\in T_{\left( 0,m \right)}$.
\end{proof}

    \begin{mylemma}\label{m mod 4 =2}
    Let $m\equiv 2 \pmod{4} \geqslant 8$ and $v=4^{\frac{m+2}{2}}+1$. Then $$
    	\left\{ av:1\leqslant a\leqslant 4^{\frac{m-4}{2}} \right\} \subset T_{(0,m)}.$$
    	
    \end{mylemma}

         \begin{proof}
    In this case, by Lemma \ref{gcd(v,n)=1}, we can easily get that $\gcd(v,n)=1$. Suppose that $a=4^{\frac{m-4}{2}}$, $w_4\left( av \right) =w_4\left( v \right) =2$, so $av\subset T_{(0,m)}$. Suppose $1\leqslant a\leqslant 4^{\frac{m-4}{2}}-1$, write $a=\sum_{i=0}^{\frac{m-6}{2}}{a_i4^i}$, where $a_i\in \left\{ 0,1,2,3 \right\} $.
    So we have
\begin{align}
		\nonumber  av&=\left(\sum_{i=0}^{\frac{m-6}{2}}{a_i4^i} \right)  \left( 4^{\frac{m+2}{2}}+1 \right) 
	=\sum_{i=0}^{\frac{m-6}{2}}{a_i\left( 4^i+4^{i+\frac{m+2}{2}} \right)}
	=\sum_{i=0}^{\frac{m-6}{2}}{a_i4^i}+a_i4^{i+\frac{m+2}{2}}.
\end{align}
It is clear that
\begin{align}
	\nonumber		w_4\left( av \right) &=
	w_4\left( \sum_{i=0}^{\frac{m-6}{2}}{a_i4^{i+\frac{m+2}{2}}}+\sum_{i=0}^{\frac{m-6}{2}}{a_i4^i} \right) 
	=w_4\left( \sum_{i=0}^{\frac{m-6}{2}}{a_i4^{i+\frac{m+2}{2}}} \right) +w_4\left( \sum_{i=0}^{\frac{m-6}{2}}{a_i4^i} \right) 
	=2w_4\left( a \right) ,
\end{align}  so $av\in T_{\left( 0,m \right)}$.\end{proof}

	\begin{mylemma} \label{m mod 4 =3}
	Let $m\equiv 3 \pmod{4}\geqslant 7$ and $v=4^{\frac{m+1}{2}}+1$. Then  $$
		\left\{ av:1\leqslant a\leqslant 4^{\frac{m-3}{2}} \right\} \subset T_{(0,m)}. $$
	   \end{mylemma}	
        	\begin{proof}
		In this case, by Lemma \ref{gcd(v,n)=1}, we can obtain that $\gcd(v,n)=1$. Suppose that $a=4^{\frac{m-3}{2}}$, $w_4\left( av \right) =w_4\left( v \right) =2$, so $av\in T\left( 0,m \right)	$. Suppose $1\leqslant a\leqslant 4^{\frac{m-3}{2}}-1$, write $a=\sum_{i=0}^{\frac{m-5}{2}}{a_i4^i}$, so that
\begin{align}
	\nonumber  	av&=\left( \sum_{i=0}^{\frac{m-5}{2}}{a_i4^i}\right)  \left( 4^{\frac{m+1}{2}}+1 \right) 
	=\sum_{i=0}^{\frac{m-5}{2}}{a_i\left( 4^i+4^{i+\frac{m+1}{2}} \right)}
	=
	\sum_{i=0}^{\frac{m-5}{2}}{a_i4^i}+a_i4^{i+\frac{m+1}{2}}.
\end{align}
It is clear that
\begin{align}
	\nonumber		w_4\left( av \right) &=
	w_4\left( 	\sum_{i=0}^{\frac{m-5}{2}}{a_i4^i}+a_i4^{i+\frac{m+1}{2}} \right) 
	=w_4\left( \sum_{i=0}^{\frac{m-5}{2}}{a_i4^{i+\frac{m+1}{2}}} \right) +w_4\left( \sum_{i=0}^{\frac{m-5}{2}}{a_i4^i} \right) 
	=2w_4\left( a \right) ,
\end{align}  so $av\in T_{\left( 0,m \right)}$.
\end{proof}

	\section{Parameters of Two Quaternary Cyclic Codes}\label{Pre}
	In this section, we will give the dimension and derive lower bounds on the minimum distance for the two quaternary cyclic codes defined above.
		\newtheorem{mylem}{Lemma}[section]

	\begin{mylem}\label{m is odd ,the size of defining set}
		If $m$ is odd and $m\geqslant 3$, then we have $\left| T_{(0,m)} \right|=\left| T_{(1,m)} \right|=2^{2m-1}-1$.
	\end{mylem}

   \begin{proof}
   	Identifying integers $\le n$ with $m$-tuples of elements of $Z_{4}=\{0,1,2,3\}\subset \Z$, we have
   	\[ T_{\left( 0,m \right)} =\{ \boldsymbol{x} \in (Z_{4})^m \mid \sum\limits_{i=1}^m x_i \equiv 0\pmod{2}\},\]where $Z_{4}$ is the set of four digits of the base $4$ expansion, and $(Z_{4})^m$ is the Cartesian product $m$ times.
   	The first $m-1$ elements $x_i$ being given, the value of $x_m$ is determined modulo $2$, hence it can take at most two distinct values.
   	In total we obtain $2\cdot 4^{m-1}-1$ vectors, excluding the all-zero vector.
   \end{proof}

	\begin{mylem}\label{m is even ,the size of defining set}
		
	If $m$ is even and $m\geqslant 2$, then we have $\left| T_{(0,m)} \right|=2^{2m-1}-2$ and $\left| T_{(1,m)} \right|=2^{2m-1}$.
	
   \end{mylem}

\begin{proof}
	  The proof is analogous except that this time we must exclude $n$ from the count of $T_{\left( 0,m\right) }.$
	\end{proof}

\newtheorem{theorem}{Theorem}[section]
	
Now we are ready to consider the parameter of cyclic codes $\C_{(0,m)}$ and  $\C_{(1,m)}$. We want to construct a cyclic code over $\mathbb{F} _4$, of parameters $[n,k,d]$ where the dimension $k$ is around half of the length $n$, while we want the distance $d$ as large as possible. With this motivation in mind, we construct different sets $S$ under different $m$'s to make sure that
$$
S\subset T_{\left( i,m \right)},
$$
or $$
vS\subset T_{\left( i,m \right)},
$$
where $v$ denotes an integer that is invertible modulo $n$. Since the size of $S$ determines the lower bounds of minimum distance of cyclic codes, we strive to maximize the size of $S.$

\subsection {The Case that $m$ is Odd}

	\newtheorem{mythm}{Theorem}[subsection]

\begin{mythm}\label{m is odd,dimension}
	If $m$ is odd and $m\geqslant 1$, then the cyclic codes $\C_{(0,m)}$ and $\C_{(1,m)}$ have the same dimension, i.e., $2^{2m-1}$.
\end{mythm}
\begin{proof}
		By Lemma \ref{m is odd ,the size of defining set}, the size of the defining set follows. Then the dimension of cyclic code is equal to
	$n-\left| T_{\left( i,m \right)} \right|$.
\end{proof}

\begin{mythm}\label{m is odd,  a pair of odd-like duadic codes}
	If $m$ is odd, then the cyclic codes $\C_{\left( 0,m \right)}$
	 and $\C_{\left( 1,m \right)}$ form a pair of odd-like duadic codes.
\end{mythm}

\begin{proof}
	For any $1\leqslant i\leqslant n-1$, where $n=4^m-1$. If we write $i$ as $i=\sum_{i=0}^{m-1}{a_i4^i},$ then $n-i=\sum_{i=0}^{m-1}{\left( 3-a_i \right) 4^i}.$
	So we have
     $ w_4\left( n-i \right) =3m-w_4\left( x \right).$ Since $m$ is odd, if $x\in T_{(0,m)},$ then $n-x\in T_{(1,m)} .$ Also, if $x\in T_{(1,m)},$ then $n-x\in T_{(0,m)} .$ Therefore we have the splitting $
     \mu =-1$ such that $$
     T_{(0,m)} =-T_{(1,m)}$$ and $$
     T_{(1,m)} =-T_{(0,m)}.$$
      From Fact \ref{duadic code definition}, the defining sets of cyclic codes $\C_{(0,m)}$ and $\C_{(1,m)}$, $T_{(0,m)}$ and $T_{(1,m)}$ partition $\mathbb{Z}_n\backslash\{0\}$. Also from the defining set, we know that $
     g_{\left( i,m \right)}\left( 1 \right) \ne 0 $.
    It then follows that $\C_{(0,m)}$ and $\C_{(1,m)}$ form a pair of duadic codes with the same parameters.
\end{proof}

\begin{mythm}
	If $m$ is odd and $m\geqslant 5$, then the cyclic codes $\C_{(0,m)}$ and $\C_{(1,m)}$ have the parameter $
	\left[ 4^{m}-1,2^{2m-1},d \right] $, where
$$
d\geqslant \begin{cases}
	4^{\frac{m-1}{2}}+1~~~~ {\rm{if}}~m\equiv 1 \pmod {4},\\
	4^{\frac{m-3}{2}}+1~~~~{\rm{if}}~m\equiv 3 \pmod {4}.\\
\end{cases}
$$

\end{mythm}

\begin{proof}
	By Theorem \ref{m is odd,dimension}, we already know the dimension of cyclic codes. We now prove the lower bounds on the minimum distance
	of the code $\C_{(0,m)}$ and $\C_{(1,m)}$. Since they are a pair of duadic codes, so we only need to get the parameters of one of the two.
	
	In the case that $m\equiv 1 \pmod{4} $ and $m\geqslant 5$, let $v=4^{\frac{m-1}{2}}+1$. It follows by Lemma \ref{m mod 4 =1} that $\gcd(v, n)=1$.
	Let $v^{-1}$ be the multiplicative inverse of $v$ modulo $n$ and let
	$\gamma =\alpha v^{-1}$. Then $\gamma $ is also an $n$-th primitive root of unity.
	It follows again by Lemma \ref{m mod 4 =1} that the defining set of
	$\C_{(0,m)}$ with respect to $\gamma $ contains the set of consecutive indices $
	\left\{ 1,2,...,4^{\frac{m-1}{2}} \right\}
	$.	The desired lower bound on $d$ then follows by the BCH
	bound on cyclic codes.
	
In the case that $m\equiv 3 \pmod{4} $ and $m\geqslant 7$, let $v=4^{\frac{m+1}{2}}+1$. It follows by Lemma \ref{m mod 4 =3} that $\gcd(v, n)=1$.
Let $v^{-1}$ be the multiplicative inverse of $v$ modulo $n$ and let
$\gamma =\alpha v^{-1}$. Then $\gamma $ is also an $n$-th primitive root of unity.
It follows again by Lemma \ref{m mod 4 =3} that the defining set of
$\C_{(0,m)}$ with respect to $\gamma $ contains the set of consecutive indices $
\left\{ 1,2,...,4^{\frac{m-3}{2}} \right\}
$.	The desired lower bound on $d$ then follows from the BCH
bound on the cyclic codes.\end{proof}
	
Note that asymptotically on $m$ our bound is of order $\frac{\sqrt{n}}{2}$ (when $m \equiv 1 \pmod{4}$) or $\frac{\sqrt{n}}{2\sqrt{2}}$ (when $m \equiv 3 \pmod{4}$) while the bound of Fact 2.5 is of order $\sqrt{n}.$ However, for general duadic codes, the odd-like distance does not coincide with the minimum distance.

\subsection {The Case that $m$ is Even}

\newtheorem{lem}{Lemma}[subsection]
\begin{mythm}\label{m even diemension}
	If $m$ is even, then the dimension for $\C_{(0,m)}$ is $2^{2m-1}+1$ and  the dimension for $\C_{(1,m)}$ is $2^{2m-1}-1$.
\end{mythm}

\begin{proof}
	By Lemma \ref{m is even ,the size of defining set}, the size of the defining set follows. Then the dimension of this cyclic code is equal to
	$n-\left| T_{\left( i,m \right)} \right|$.
\end{proof}

\begin{mythm}
	If $m$ is even and $m\geqslant 8$, then the cyclic code  $\C_{\left( 0,m \right)}$ have the parameter $\left[ 4^{m}-1,2^{2m-1}+1,d \right]$, where
$$
d\geqslant \begin{cases}
	4^{\frac{m-4}{2}}+1~~~~{\rm{if}}~m\equiv 2 \pmod{4},\\
	4^{\frac{m-4}{2}}+1~~~~{\rm{if}}~m\equiv 4 \pmod{8}. \\
\end{cases}
$$

\end{mythm}

\begin{proof}
		From Theorem \ref{m even diemension}, we already know the dimension of this cyclic code. We now prove the lower bounds on the minimum distance
		of the code $\C_{(0,m)}$. In the case that $m\equiv 2 \pmod{4} $ and $m\geqslant 8$, let $v=4^{\frac{m+2}{2}}+1$. It follows by Lemma
		\ref{m mod 4 =2} that $\gcd(v, n)=1$.
		Let $v^{-1}$ be the multiplicative inverse of $v$ modulo $n$ and let
		$\gamma =\alpha v^{-1}$. Then $\gamma $ is also an $n$-th primitive root of unity.
		It follows again by Lemma \ref{m mod 4 =2} that the defining set of $\C_{(0,m)}$ with respect to $\gamma $ contains the set of consecutive indices  $
		\left\{ 1,2,...,4^{\frac{m-4}{2}} \right\}
		$.
		In the case that $m\equiv 4 \pmod{8}$, it follows by Lemma
		\ref{m=0(mod4)} that $\gcd(v, n)=1$.
		Let $v^{-1}$ be the multiplicative inverse of $v$ modulo $n$ and let
		$\gamma =\alpha v^{-1}$. Then $\gamma $ is also an $n$-th primitive root of unity.
		It follows again by Lemma \ref{m=0(mod4)} that the defining set of $\C_{(0,m)}$ with respect to $\gamma $ contains the set $
		\left\{ 1,2,...,4^{\frac{m-4}{2}} \right\}
		$. The desired lower bound on $d$ follows then by the BCH
		bound on the cyclic codes.\end{proof}

Note that asymptotically on $m$ our bound is of order $\frac{\sqrt{n}}{4}.$ When referring to the parameter of $\C_{(1,m)}$ when $m$ is even. We found some relationship between 2-weights in \cite{TD} and 4-weights in our paper.

\begin{lem} \label{W2 and W4}

\begin{equation}
	w_4\left( 2a \right) \equiv w_4\left( a \right) \pmod {2}
	\Longleftrightarrow w_2\left( a \right) \equiv 0 \pmod {2}    \tag{1}
\end{equation}
\begin{equation}
	w_4\left( 2a \right) \equiv w_4\left( a \right) +1 \pmod {2} \Longleftrightarrow w_2\left( a \right) \equiv 1 \pmod {2}    \tag{2}
\end{equation}

\end{lem}

\begin{proof}
	First let us recall the definition of 2-weights and 4-weights. For a positive integer $a$, if we write $a=\sum_{i=0}^{t}{b_i4^i}$
	or $a=\sum_{i=0}^{s}{c_i2^i}.$ The 4-weights of $a$ means $w_4=\sum_{i=0}^{t}{b_i}$, and 2-weights of $a$ means $w_2=\sum_{i=0}^{s}{c_i}$.
	
	We will show the Equation (1), and the Equation (2) is similar. For any positive integer $a$, since $b_i \in \Z_4$, so we can write $b_i=2d_i+e_i$, where $d_i,e_i\in\Z_2$.
	\begin{align}
	   a =\sum_{i=0}^t{b_i4^i}   \notag   \notag
		 =\sum_{i=0}^t{b_i2^{2i}}   \notag
		 =\sum_{i=0}^t{\left( 2d_i+e_i \right) 2^{2i}}   \notag
		 =\sum_{i=0}^t{2d_i2^{2i}}+e_i2^{2i}, w_4\left( 2a \right) =\sum_{i=0}^t{\left( d_i+2e_i \right)} \end{align}and $
w_4\left( a \right) =\sum_{i=0}^t{\left( 2d_i+e_i \right)}.
$ Then it is clearly that  $$
w_2\left( a \right) \equiv 0\left( \mathrm{mod}\;2 \right) \Leftrightarrow \sum_{i=0}^t{e_i}=\sum_{i=0}^t{d_i}\Longleftrightarrow w_2\left( a \right) \equiv w_4\left( a \right) \left( \mathrm{mod}\;2 \right).
$$
This completes the proof.\end{proof}
	Getting suitable $v$ directly is very hard, we must ensure that when $m$ is even, $\gcd(v,n)=1$ and for some  consecutive integer $a$ we have $\left\{ av: y\leqslant a\leqslant x \right\} \subseteq T_{\left( 1,m \right)}$. In the meanwhile, for the convenience of calculating $av$ modulo $n$, we want $xv<n$. By applying Lemma \ref{W2 and W4}, maybe we can simplify this question.
	In order to distinguish the different defining sets, we let
	$$
	T_{2\left( 0,m \right)}=\left\{ 1\leqslant j\leqslant n-1: w_2\left( j \right)~\rm{is}~even \right\},
	$$
	 $$T_{2\left( 1,m \right)}=\left\{ 1\leqslant j\leqslant n-1: w_2\left( j \right)~\rm{is}~odd \right\},$$
	 $$
	 T_{4\left( 0,m \right)}=\left\{ 1\leqslant j\leqslant n-1: w_4\left( j \right)~\rm{is}~even \right\},
	 $$
	 $$
	 T_{4\left( 1,m \right)}=\left\{ 1\leqslant j\leqslant n-1: w_4\left( j \right)~\rm{is}~odd \right\}.
	 $$
From \cite{TD}, we can easily get the suitable $v_1$ and consecutive $a$ such that $\left\{ av_1: y\leqslant a\leqslant x \right\} \subset T_{2\left( 1,m \right)}$. Then a suitable $v_2$ can be obtained easily by applying method in this paper such that $
\left\{ av_1v_2: y\leqslant a\leqslant x \right\} \subset T_{4\left( 0,m \right)}$. If $w_2(av_1v_2)\in T_{2(1,m)}$, from the results of Lemma \ref{W2 and W4}, we have
$$
w_4\left( 2av_1v_2 \right) \equiv w_4\left( av_1v_2 \right) +1 \pmod{2},
$$
so that $\left\{ 2av_1v_2: y\leqslant a\leqslant x \right\} \subset T_{4\left( 1,m \right)}.$

Hence, the following open problem is very
natural.
{\bf Open Problem 1:} By applying lemma above, how to estimate the parameters of the cyclic code $\C_{(1,m)}$ when $m$ is even?

Here  is a partial answer to that question.

\begin {mythm}  \label{partial proof}
 Suppose $m$ is even, and $ m\equiv 14 \pmod{16}$, $ m\geqslant 30$. Let the parameter of $\C_{(1,m)}$ be $[n,k,d]$, then $d\geqslant2^{\frac{m-14}{4}}+1$.

\end{mythm}

\begin{proof}
	Let $v_1=2^{\frac{m-4}{2}}-1$, $v_2=4^{\frac{m+2}{8}}+1$. In this case, $\gcd\left( n_4,v_1 \right) =\gcd\left( n_4,v_2 \right) =1$. We have
	$$
	\left\{ av_1v_2:1\leqslant a\leqslant 2^{\frac{m-14}{4}} \right\} \subset T_{4\left( 0,m \right)},
	$$
	$$
	\left\{ av_1v_2:1\leqslant a\leqslant 2^{\frac{m-14}{4}} \right\} \subset T_{2\left( 1,m \right)}.
	$$
	Since $
	w_2\left( av_1v_2 \right) \equiv 1 \pmod{2}
	$, so $$
	w_4\left( 2av_1v_2 \right) \equiv w_4\left( av_1v_2 \right) +1 \pmod {2}.$$ Therefore,  $
	\left\{ 2av_1v_2:1\leqslant a\leqslant 2^{\frac{m-14}{4}} \right\} \subset T_{4\left( 1,m \right)}.
	$\end{proof}

\section{Two Dual Cyclic Codes}
	\newtheorem{mythm2}{Theorem}[section]
   According to the defining set, we note that $T_{(0,m)}$ and $T_{(1,m)}$  partition ${\mathbb{Z}}_{n}\backslash\{0\}$. So it is very clear that
\begin{enumerate}
	\item [(1)] If $m$ is even, then we have $\C_{\left( 0,m \right)}^{\bot}$ is the even-weight subcode of $\C_{(1,m)}$, $\C_{\left( 1,m \right)}^{\bot}$ is the even-weight subcode of $\C_{(0,m)}$.
	
	\item [(2)] If $m$ is odd, then we have $\C_{\left( 0,m \right)}^{\bot}$ is the even-weight subcode of $\C_{(0,m)}$, $\C_{\left( 1,m \right)}^{\bot}$ is the even-weight subcode of $\C_{(1,m)}$.
	\end{enumerate}
Therefore, we have the following theorems:
\begin{mythm2}
	If $m$ is odd and $m\geqslant 3$, then $\C_{\left( 0,m \right)}^{\bot}$ and $\C_{\left( 1,m \right)}^{\bot}$ form a pair of even-like duadic codes with parameters $\left[ 4^m-1,2^{2m-1}-1,d \right] $, where
	$$
	d\geqslant \begin{cases}
		4^{\frac{m-1}{2}}+2~~~~{\rm{if} }~m\equiv 1 \pmod {4},\\
		4^{\frac{m-3}{2}}+2~~~~{\rm{if}}~m\equiv 3 \pmod {4}.\\
	\end{cases}
	$$
	
\end{mythm2}

\begin{mythm2}
		If $m$ is even and $m\geqslant 6$, then the cyclic code  $\C_{\left( 1,m \right)}^{\bot}$ have the parameter $\left[ 4^{m}-1,2^{2m-1},d \right]$, where
$$
d\geqslant \begin{cases}
	4^{\frac{m-4}{2}}+2~~~~{\rm{if}}~m\equiv 2 \pmod{4},\\
	4^{\frac{m-4}{2}}+2~~~~{\rm{if}}~m\equiv 4 \pmod{8}. \\
\end{cases}
$$
	
\end{mythm2}

\begin{mythm2}
	If $m$ is even, then the cyclic codes $\C_{(0,m)}$ and $\C_{(1,m)}$ are both LCD codes.
	\end{mythm2}
\begin{proof}
	Let $\alpha$ be a generator of $GF(4^m)^*$.
	Let $i=0$ or $1$. Then  $\C_{(i,m)}$ has generator polynomial
	$$
	g_{\left( i,m \right)}\left( x \right) =\prod_{\begin{array}{c}
			1\leqslant j\leqslant n-1\\
			w_4\left( j \right) \equiv i\pmod{2}\\
	\end{array}}^{}{\left( x-\alpha ^j \right).}
	$$
Then $\C_{\left( 0,m \right)}\cap \C_{\left( 1,m \right)}$ has generator polynomial $$
g\left( x \right) =lcm\left\{ g_{\left( 0,m \right)}\left( x \right) ,g_{\left( 1,m \right)}\left( x \right) \right\} =\frac{x^n-1}{x-1}=x^{n-1}+x^{n-2}+\cdots +x+1.
$$
Therefore, $\C_{\left( 0,m \right)}\cap \C_{\left( 1,m \right)}$ is a repetition code $<\bf 1>$. Since when $m$ is even, $\C_{\left( 0,m \right)}^{\bot}$ is a subcode of $\C_{(1,m)}$, so we have
$$
\C_{\left( 0,m \right)}\cap \C_{\left( 0,m \right)}^{\bot}\subset \C_{\left( 0,m \right)}\cap \C_{\left( 1,m \right)}=<\bf 1>.
$$
Also we notice that $\C_{\left( 0,m \right)}^{\bot}$ is even-weight, so we obtain $$
\C_{\left( 0,m \right)}\cap \C_{\left( 0,m \right)}^{\bot}=\left\{ 0 \right\}.
$$
Therefore, $\C_{(0,m)}$ is a LCD code and the same condition is true for $\C_{(1,m)}$.\end{proof}

\begin{ex}
	Let $m=2$. Then $\C_{(0,m)}$ has parameters $[15,9,3]$ and $\C_{\left( 0,m \right)}^{\bot}$ has parameters $[15,6,6]$. The code $\C_{(1,m)}$ has parameters $[15,7,5]$ and $\C_{\left( 1,m \right)}^{\bot}$ has parameters $[15,8,4]$. In this case, the cyclic codes $\C_{(0,m)}$ and $\C_{(1,m)}$ are both LCD codes.\end{ex}

\section{Extension Codes and  Subfield Subcodes}
 \subsection{Extension Codes}
 	\newtheorem{mythm3}{Theorem}[subsection]
 	\newtheorem{mylem3}{Lemma}[subsection]
\begin{mythm3}(\cite[Theorem 6.4.12]{HufPle})
	Let $\D_{1}$ and $\D_{2}$ be a pair of odd-like duadic codes of length $n$ over $\mathbb{F}_q$. Let $\overline{\D_i}$
	 stand for the extension code of
	$\D_{i}$  for each $i$.
	 Assume that there is a solution $\gamma\in\mathbb{F}_q$ to the equation
	$$
	1+\gamma ^2n=0,
	$$
	then the following hold:
	\begin{enumerate}
		\item [(1)] If $\mu _{-1}$ gives the splitting for $\D_{1}$ and $\D_{2}$, then
		$\overline{\D_1}$ and $ \overline{\D_2}$ are self-dual.
		\item [(2)] If $\D_1\mu _{-1}=\D_2$, then $\overline{\D_1}$ and $ \overline{\D_2}$ are duals of each other.
		\end{enumerate}
	
   \end{mythm3}

When $m$ is odd, take $\gamma=1\in\mathbb{F}_4$, then we know that the extension code of $\C_{(i,m)}$ is self-dual for each $i$. Now we consider the type of extension codes.

\begin{mylem3}(\cite[Proposition 3.1]{type II}) If $\C$ is self-orthogonal so is
	$\phi\left( \C\right) $. In this case $\phi\left( \C\right) $ is
	a Type I (resp. Type II) code if and only if  $\C$ is a Type \Rmnum{1} (resp. Type II) code.
	
\end{mylem3}

\begin{mylem3}(\cite[Theorem 4.3]{type II}) Let $\C$ be an odd-like $Q$-code over $\mathbb{F}_4$ of length $n\equiv 3$ modulo 4, with splitting given by $\mu_{-1}$. Let $\bar{\C}$ be the extended code of $\C$. Then the Gray image of $\bar{\C}$ is Type II.

\end{mylem3}

\begin {mylem3}(\cite[Theorem 16]{Q-codes})
   Every self-dual or strictly self-dual extended cyclic quaternary code is an extended Q-code.
\end{mylem3}
By applying the three above lemmas, we can directly get that:
\begin{mythm3}
	If $m$ is odd, then the extension codes  $\overline{\C_{\left( 0,m \right)}}$ and $\overline{\C_{\left( 1,m \right)}}$ of cyclic codes $\C_{(0,m)}$ and $\C_{(1,m)}$ are Type II self-dual.
\end{mythm3}

\subsection{Subfield Subcode and Trace Code}

\begin{mythm3}(\cite[Theorem 6.3.5 Subfield subcode]{Ling san})
	Let $\C$ be an $[N,K,D]$-linear code over $\mathbb{F}_{q^m}$. Then the subfield subcode $\C|\mathbb{F}_q:=\C\cap \mathbb{F}_{q}^{N}$ is an $[n,k,d]$-linear code over $\mathbb{F}_q$ with $n=N$, $
	k\geqslant mK-\left( m-1 \right) N$ and $d\geqslant D$. Moreover, an
	$[N,mK-(m-1)N,D]$-linear code over $\mathbb{F}_q$ can be provided that $mK>(m-1)N$.

	\end{mythm3}

	\newtheorem{cor1}{Corollary}[subsection]
\begin{cor1}

   	If $m$ is odd, then a subfield subcode $\C|\mathbb{F}_q:=\C_{(i,m)}\cap \mathbb{F}_{q}^{N}$ is an $[n,k,d]$-linear code over $\mathbb{F}_q$ with $n=4^m-1$, $k\geqslant 1$ and $d\geqslant 4^{\frac{m-3}{2}}+1$ for each $i=0,1$.	
\end{cor1}
It is interesting to compare the parameters of the codes in \cite{TD} and the parameters of subfield subcodes in this paper. Here we list a good example. And we let $\C_{t(0,m)}$ and $\C_{t(1,m)}$ denote the trace codes of $\C_{(0,m)}$ and $\C_{(1,m)}$, respectively. Let $\overline{\C_{t(0,m)}}$ and $\overline{\C_{t(1,m)}}$ denote the extension codes of $\C_{t(0,m)}$ and $\C_{t(1,m)}$, respectively.

	\newtheorem{myex3}{Example}[subsection]
\begin{myex3}\label{example}
	When $m=3$, the parameters of $\C_{t(0,m)}$ and $\C_{t(1,m)}$ are $[63,48,5]$, which is near optimal. And the parameters of $\overline{\C_{t(0,m)}}$ and $\overline{\C_{t(1,m)}}$ are $[64,48,6]$, which is an optimal code. For binary codes in \cite{TD} with length $n=63$, we have the parameters
	$[63,31,6]$ and $[63,33,7]$. From the code table in \cite{code table}, we can see they are not as good as  $\C_{t(0,m)}$ and $\C_{t(1,m)}$.
\end{myex3}

\section{Conclusion}
In this paper, we have introduced a quaternary analogue of the Tang-Ding codes in lengths $4^m-1.$ For $m$ odd or $m \equiv 2 \pmod{4}$ or $m \equiv 4 \pmod{8}$, we have given a lower bound on the minimum distance of order the square root of the length. The main open problem is to understand better the parameters of the code $\C_{(1,m)}$ if $m$ is even. A partial answer to that question is Theorem \ref{partial proof}.
On a more positive side, Example \ref{example} shows that the parameters of the Trace code can outperform those of the Tang-Ding codes. It would therefore be desirable
to derive sharper bounds on the minimum distance of the Trace codes.

\end{document}